\documentclass{article}%D
\usepackage[utf8]{inputenc}
\usepackage{lineno}%D
\usepackage{graphicx,xcolor}
\usepackage{mathtools}
\usepackage{amsmath,amssymb,amsthm,xfrac,bm,bbm}
\usepackage{enumitem}
\usepackage{natbib}

\newcommand{\st}{\,\middle|\,}

\newcommand{\pb}[1]{\mathbb P\left[#1\right]}
\newcommand{\expect}[1]{\mathbb E\left[#1\right]}
\newcommand{\var}[1]{\mathbb V\left[#1\right]}

\newcommand{\expectVariable}[2]{\mathbb E_{#1}\left[#2\right]}

\newtheorem{definition}{Definition}

\newtheorem{theorem}{Theorem}
\newtheorem{lemma}{Lemma}
\newtheorem{corollary}{Corollary}

\newcommand{\boxA}{\textsc{Box~A}}
\newcommand{\boxB}{\textsc{Box~B}}
\newcommand{\onebox}{\textsc{1-box}}
\newcommand{\twobox}{\textsc{2-box}}
\newcommand{\oneboxer}{\textbf{1-boxer}}

\newcommand{\OM}{\textbf{Omega}}
\newcommand{\alice}{\textbf{Alice}}
\newcommand{\emma}{\textbf{Emma}}
\newcommand{\dan}{\textbf{Dan}}
\newcommand{\aixi}{\textbf{AIXI}}

\newcommand{\boxn}[1]{\textsc{Box-}#1}

\usepackage{todonotes}

\makeatletter
\newtheorem*{rep@theorem}{\rep@title}
\newcommand{\newreptheorem}[2]{%
\newenvironment{rep#1}[1]{%
 \def\rep@title{#2 \ref{##1}}%
 \begin{rep@theorem}}%
 {\end{rep@theorem}}}
\makeatother

\newreptheorem{theorem}{Theorem}

\title{Purely Bayesian counterfactuals \\ versus Newcomb's paradox}
\author{Lê Nguyên Hoang \\ EPFL}
\date{}

\begin{document}%D
% \pagewiselinenumbers
\maketitle

% Abstract word limit: 200
% Article word limit: 8,000

\begin{abstract}
This paper proposes a careful separation between an entity's {\it epistemic system} and their {\it decision system}. Crucially, Bayesian counterfactuals are estimated by the epistemic system; not by the decision system. Based on this remark, I prove the existence of Newcomb-like problems for which an epistemic system necessarily expects the entity to make a counterfactually {\it bad} decision.

I then address (a slight generalization of) Newcomb's paradox. I solve the specific case where the player believes that the predictor applies Bayes rule with a supset of all the data available to the player. I prove that the counterfactual optimality of the \onebox{} strategy depends on the player's prior on the predictor's additional data. If these additional data are not expected to reduce sufficiently the predictor's uncertainty on the player's decision, then the player's epistemic system will counterfactually prefer to \twobox{}. But if the predictor's data is believed to make them quasi-omniscient, then \onebox{} will be counterfactually prefered. Implications of the analysis are then discussed.

More generally, I argue that, to better understand or design an entity, it is useful to clearly separate the entity's epistemic, decision, but also {\it data collection}, {\it reward} and {\it maintenance} systems, whether the entity is human, algorithmic or institutional.
\end{abstract}

\section{Introduction}

Newcomb's paradox is an iconic paradox of decision theory. Introduced by \cite{nozick69}, the problem involves a player, call her \alice{}, and a predictor that we will name \OM{}. \OM{} predicts \alice{}'s behavior, and, essentially, determines her reward based on this prediction. This will then put \alice{} into some seemingly impossible dilemma.

On one hand, a seemingly {\it causal} argument suggests that \alice{} should ignore her predictability, especially once \OM{}'s prediction has {\it already} been made. This argument suggests that \alice{} should then adopt a strategy called \twobox{}.

On the other hand, by modifying her behavior, and assuming that \OM{}'s prediction is really reliable, \alice{} seems able to bias \OM{}'s prediction so that it becomes aligned with \alice{}'s interest. This suggests that \alice{} may be able to hack her predictability to gain more rewards — a strategy known as \onebox{}. More thorough details are provided in Section \ref{sec:problem}.

Remarkably, scholars are extremely divided on what strategy \alice{} should adopt. In fact, Newcomb's paradox seems to be unveiling a fundamental gap in our understanding of {\it decision theory}, and thus of related topics such as free will, game theory and algorithm design. Because of this, over the last half a century, Newcomb's paradox has received a lot of attention from philosophers (\cite{schlesinger74,locke78,horwich85,nozick94,nozick97}, to name a few), mathematicians (\cite{gardner74,wijayatunga19,giacopelli19}), economists (\cite{broome89,sugden91,weber16}), political scientists (\cite{brams75,frydman82}), computer scientists (\cite{aaronson13,everitt18}) and psychologists (\cite{bar72}). 

In addition to connections to fundamental philosophical problems, Newcomb's paradox has been linked to practical problems, such as the voting dilemma (see \cite{elster87}). Assuming that \alice{}'s vote has a negligible effect, but is time-costly for her, intuitively, the equivalent of \twobox{} would be to argue that \alice{}'s optimal selfish strategy is to {\it not} take the time to vote. However, it has been argued that, by going voting, \alice{} changes what \alice{} would predict about the behaviour of individuals similar to her, thereby significantly increasing the probability that \alice{}'s favorite candidate will get elected — which is in \alice{}'s selfish interest. A \onebox{} argument suggests that \alice{} should vote.

% Many previous researches relied either on the researcher's intuition, or on a specific game-theoretical modeling of the paradox. However, such models are arguably {\it high level}. They arguably do not build upon first principles.
% of epistemology and decision principles as solid as {\it purely Bayesian counterfactual optimization}.
In this paper, we present a novel analysis of Newcomb's paradox. The core idea of the analysis is a clear separation between an individual's {\it epistemic system}, and their {\it decision system}. The epistemic system will be assumed to be {\it purely Bayesian}, which means that all of its thoughts result from data and the laws of probability. In particular, thereby, the epistemic system can engage in counterfactual reasoning, to determine what decisions yield the largest counterfactual expected rewards.

As any good Bayesian, at any point in time, the epistemic system must also assign a probability to any future event. This evidently includes the individual's future decision. Now, in principle, such a probability could take the values $0$ or $1$. However, we stress the fact that an epistemic system with such a perfect knowledge of the individual's decision {\it cannot} engage in a counterfactual reasoning. Under perfect knowledge, and under the laws of probability, Bayesian counterfactual reasoning becomes nonsensical.

As a result, most of this paper discusses the arguably more realistic case of {\it imperfect} knowledge of the decision system. This means that the epistemic system fails to know with full certainty the algorithm executed by the decision system, and the decision system's inputs. In such a case, in Section \ref{sec:impossibility}, we prove Theorem \ref{th:impossibility}, which says that no decision system can be expected to guarantee counterfactual optimization. We then discuss consequences of the theorem for decision theory.

In Section \ref{sec:newcomb}, we then tackle a generalization of Newcomb's paradox, under the assumption that a Bayesian \OM{} has collected a supset of \alice{}'s data. We show that the counterfactual optimality of a decision strongly depends on the extent to which \alice{} believes the predictor \OM{} to know more than her about her decision system. Indeed, if \alice{} does not believe that \OM{} knows more than her, then \twobox{} is the only counterfactually optimal decision. But if \alice{} expects \OM{} to know a lot more to the point of being quasi-omniscient, then \onebox{} becomes the only counterfactually optimal decision. These findings are formalized by Theorem \ref{th:main}, and by its diverse corollaries (see Section \ref{sec:corollaries}).

Section \ref{sec:discussion} will then discuss our results, and raise further questions. First, we analyze the impact of relaxing the assumptions of the previous section. Second, we discuss implications of our analysis to practical Newcomb-like problems. Third, we generalize the main idea of the paper, namely the separation of epistemic and decision systems, to a further separation of different systems of any information processing entity. 

Finally, Section \ref{sec:conclusion} concludes.

\section{Counterfactual optimization is impossible}
\label{sec:impossibility}

In this section, we show that no entity can guarantee counterfactual optimization. To understand this claim, we first insist on the distinction between an entity's {\it epistemic system} and their {\it decision system}. We then stress the fact that counterfactual optimization is not an {\it algorithm}, but a {\it property}, which results from the interaction between the epistemic system, the decision system and the problem at hand. We then state and prove the impossibility of counterfactual optimization, and discuss consequences. 

\subsection{Epistemic versus decision systems}
\label{sec:epistemic_decision}

In this paper, we study how an information processing entity, like a human, a machine or an organization, infers a world view and makes a decision. Confusingly, these two different tasks seem entangled. On one hand, it seems important to first infer a world view before making a decision. But as highlighted by Newcomb-like paradoxes, it seems that a decision can sometimes inform a world view.

To disentangle views and decisions, we propose to clarify the distinction between these two tasks. In fact, we will assume that each task is handled by a specific system of the entity. % The {\it epistemic system} will infer a world view, while the {\it decision system} will output decisions. 
More specifically, in this paper, we focus on an entity that we call \alice{}, which possesses both an epistemic and a decision system. We call them respectively \emma{} and \dan{}. \emma{} will infer the state of the world based on her data. \dan{} will compute decisions based on his data.

\subsubsection{Bayesianism}

Except in Section \ref{sec:non-bayesian}, we assume \emma{} to be a {\it pure Bayesian}. In other words, \emma{} will apply {\it solely} the laws of probabilities to compute credences and expectations\footnote{This contrasts with some versions of causal decision theory that invoke \cite{pearl18}'s ``do-operator''. This operator is thus {\it added} to Bayesianism, and also depends on some causal modeling of the problem. As a result, we do not consider it to be {\it purely} Bayesian.}. In particular, \emma{}'s credence in a theory $T$ of the state of the world, given the data $D_{\bf E}$ she collected, will be given by Bayes rule:
\begin{equation}
  \pb{T \st D_{\bf E}} = \frac{\pb{D_{\bf E} \st T} \pb{T}}{\pb{D_{\bf E}}}.
\end{equation}
Note that we adopt here the purely Bayesian interpretation of probabilities. Namely, as discussed by \cite{laplace1840}, such probabilities describe an entity's knowledge and uncertainty, and apply even in a deterministic universe. 

In this paper, we will discuss only \emma{} and \OM{}'s credences, mostly under the assumption that they have common priors, and that \emma{}'s data $D_{\bf E}$ are common knowledge. As a result, without loss of generality, we omit the conditioning by $D_{\bf E}$ in our notations.

\subsubsection{Counterfactuals}
\label{sec:counterfactual}

In our paper, we analyze \dan{}'s decision through the lenses of purely Bayesian counterfactuals. 

\begin{definition}
\label{def:counterfactual}
  \emma{} counterfactually prefers a decision $\textsc{X}$ if, for any alternative decision $\textsc{Y}$, \emma{} estimates the counterfactual expected rewards to be larger assuming $\textsc{X}$ than assuming $\textsc{Y}$, i.e.
  \begin{equation}
    \expect{\mathcal R \st \textsc{X}}
    \geq \expect{\mathcal R \st \textsc{Y}}.
  \end{equation}
  Recall that the expectations are implicitly also conditioned by \emma{}'s data $D_{\bf E}$.
\end{definition}

Counterfactual optimality is often invoked to argue in favor of one decision over the other. In the case of the COVID-19 pandemic, for instance, some observers have argued that lockdown was not actually as bad for the economy as one might expect. Their argument was not that the economy did not suffer during the lockdown; rather, they argued that the counterfactual world without a lockdown decision would have seen a similar blow to the economy. By then also taking public health into consideration, such observers argued that lockdown was very probably counterfactually optimal. In other words, these observers have essentially argued that
\begin{equation}
  \expect{\mathsf{Welfare} \st \textsc{Lockdown}} \geq \expect{\mathsf{Welfare} \st \textsc{No-Lockdown}},
\end{equation}
where $\mathsf{Welfare}$ refers to some sort of public good, which involves both public health and financial safety. Interestingly, some opponents agreed with the general approach to determine if $\textsc{Lockdown}$ was a good decision, but disagreed with the estimation of the counterfactual expectations\footnote{Note that the counterfactual expectations should be also all conditioned on all the available data at the moment of the decision lockdown, including the dangerous spread of the COVID-19 disease.}.

Counterfactual optimization is also often invoked in the analysis of Newcomb's paradox, which is presented in Section \ref{sec:problem}. Again, the disagreements often lie in the estimation of the counterfactual expectations. The \oneboxer{} will typically argue that $\expect{\mathcal R \st \onebox{}} = R$, where $R$ is the content of the opaque box if the predictor \OM{} predicts \onebox{}, while the counterfactual expected rewards are $\expect{\mathcal R \st \twobox{}} = r < R$, since \OM{} would then leave the opaque box empty. However, the \twobox{} will argue that $\expect{\mathcal R \st \twobox{}} = r + \expect{\mathcal R \st \onebox{}}$.
Clearly, both cannot be right simultaneously. Section \ref{sec:newcomb} will aim to clarify what is going on.

% In this section, we will formalize and prove Theorem \ref{th:optimization}, namely that, essentially, pure Bayesians who knowingly counterfactual optimize decide to \twobox{}. But, first, we need to resolve a technical difficulty.

\subsubsection{Imperfect knowledge}
\label{sec:imperfect}

One feature of Bayesianism is that, if $\pb{D} = 0$, then $\pb{T \st D}$ is ill-defined. In a sense, a Bayesian cannot consider events that they have completely discarded. But this raises a serious technical issue, if we assume that \emma{} knows {\it for sure} \dan{}'s decision \textsc{X}. Indeed, \emma{} would then be completely discarding the possibility that \dan{} makes an alternative decision \textsc{Y}.
In other words, \emma{} may believe $\pb{\textsc{Y}} = 0$. But then, the counterfactual expectation $\expect{\mathcal R \st \textsc{Y}}$ would be ill-defined, which makes counterfactual optimization nonsensical.
To resolve this issue, in this paper, we will only consider Bayesians with {\it imperfect knowledge}.

\begin{definition}
  A Bayesian has imperfect knowledge about $X$ if, for any possible value $x$ of the event, the Bayesian assigns a strictly positive probability to $X=x$.
\end{definition}

% In general, in pure Bayesianism, such problems are resolved by assigning a positive probability to at least one theory $T$, which itself assigns a positive probability to all data $D$. As a result, we will have $\pb{D} \geq \pb{D \wedge T} = \pb{D \st T} \pb{T} > 0$, which makes $\pb{T \st D}$ well-defined, even though we may have $\pb{D} \ll 1$. More generally, pure Bayesianism makes more sense by avoiding invoking probabilities equal to $0$ or $1$, especially about future data.

In the case of \alice{}, it seems actually reasonable to assume that, especially in practice, \emma{} cannot guarantee that \dan{} will execute a given decision algorithm. After all, even if \emma{} carefully inspected \dan{} before \dan{} makes a decision, and even if \emma{} knows exactly the data given to \dan{}, \dan{} may still slightly change before the computation is executed. In other words, \emma{} may precisely know what \dan{} was like seconds ago, but she cannot know {\it for sure} what \dan{} will be doing in a few seconds, when \dan{} {\it actually} runs his computation to deliver a decision. It thus seems reasonable to assume that \emma{} has an imperfect knowledge of \dan{}'s decision.

\subsection{Formal impossibility}
\label{sec:formal_impossibility}

Finally, we can state the main result of this section. It asserts that no entity can guarantee counterfactual optimization.

\begin{theorem}
\label{th:impossibility}
    There exists a decision problem with $n$ options, for which any entity with imperfect knowledge assigns a probability at least $1-\sfrac{1}{n}$ to counterfactually bad decisions.
\end{theorem}

\begin{proof}
  Consider $n$ opaque boxes. \dan{} must choose one of the boxes. Denote $\boxn{i}$ the fact that \dan{} chooses the $i$-th box.

  But, now assume that \emma{} gets convinced that a predictor \OM{} has the same data as her and knows her prior. As a result, \emma{} believes with probability 1 that \OM{} can compute \emma{}'s credence $\pb{\boxn{i}}$ in \dan{} deciding $\boxn{i}$.
  Suppose also that \emma{} also believes with probability 1
  % that \OM{} decides what reward $R_i$ to put in box $i$, based only on the values of the credences $\pb{\boxn{j}}$ of what \dan{} will do.
  % Assume also that \emma{} knows \OM{}'s decision algorithm. Then, given her prior \emma{} effectively knows the rewards $R_i$'s inside the boxes. In particular, it is then straightforward for her to determine what is counterfactually preferrable.
  that \OM{} decides the content of box $i$ as follows. \OM{} computes the smallest value $*$ of $i$ such that $\pb{\boxn{i}} \leq 1/n$. \OM{} then sets $R_* = 1$, and $R_j = 0$ for $j \neq *$.

  Then, from \emma{}'s perspective, given that she knows the rewards $R_i$'s, any decision $\boxn{j}$ different from $\boxn{*}$ is counterfactually bad.
  However, we also know that \emma{} also assigns a probability at most $1/n$ to $\boxn{*}$. Thus, according to \emma{}, there is a probability at least $1-\sfrac{1}{n}$ that \dan{}'s decision is counterfactually bad.
\end{proof}

Note that our proof assumes that \emma{} knows {\it for sure} that \OM{} knows \emma{}'s data and prior. Interestingly, it is however robust to relaxing this condition, and to assuming that \emma{} only strongly believes that \OM{} knows her data and prior. This then yields a probability of at least $1-\sfrac{1}{n}-o(1)$ of a counterfactually bad decision, where $o(1)$ denotes a term arbitrarily small by considering that \emma{}'s credences are arbitrarily close to $1$.

% The proof of the theorem relies on the construction of a decision problem based on a sort of diagonalization. In a sense, this problem is designed so that any decision of \dan{} will go wrong. But to see why this problem indeed tricks \dan{}, we first discuss an important result from probability theory.

\subsubsection{Why AIXI escapes our impossibility theorem}
\label{sec:aixi}

It is noteworthy that Theorem \ref{th:impossibility} does not apply to the framework of \aixi{}, a counterfactually optimal decision algorithm introduced by \cite{hutter2001,hutter2004}. In this framework, an entity called \aixi{} interacts with its environment by making decisions based on its observed past data, and the environment responds to \aixi{}'s decision by providing new data and a reward\footnote{The environment is assumed to reply according to a computable probability distribution that depends on past data and decisions, and the latest decision.}.

Essentially, \aixi{} escapes our impossibility theorem because the \aixi{} framework prevents external entities from making \aixi{}'s reward depend on \aixi{}'s uncertainty about its decision.
More precisely, \aixi{}'s expected rewards given \aixi{}'s decision \textsc{X} are assumed to be independent from \aixi{}'s uncertainty about deciding \textsc{X}. Formally, this assumption enforces the equality\footnote{This is formalized by an environment $\mu$ whose outputs given the past and \aixi{}'s decision cannot depend on the policy $\pi$ of \aixi{}. This allows \aixi{}'s counterfactual rewards to be a linear function of $\pi$, where $\pi$ is regarded as a probability distribution. Conversely, our proof of Theorem \ref{th:impossibility} designed a problem where the rewards given a decision \textsc{X} highly depend on $\pi$.}
$\expect{R \st \textsc{X} \wedge \pb{X}=p} = \expect{R \st \textsc{X} \wedge \pb{X} = q}$ for any values of $p$ and $q$. As a result, \aixi{}'s credence $\pb{\textsc{X}}$ in the fact that it will decide \textsc{X} cannot be exploited by the environment to bias \aixi{}'s rewards.
This contrasts with our proof of Theorem \ref{th:impossibility} in which, if $\pb{\textsc{X}}$ is large, then the reward given \textsc{X} would be designed to be small.

% In a sense, the \aixi{} framework prevents the environment from adapting its rewards from what the environment knows about \aixi{} (and in particular about \aixi{}'s epistemic system). As a result, nothing about \aixi{} apart from its actual decision could be exploited by the environment.
Unfortunately, the restricted framework of \aixi{} has been argued to make \aixi{} unrealistic. In particular, the fact that the environment cannot exploit its understanding of \aixi{} seems incompatible with {\it embedded agency}, as discussed by \cite{demski19}. Embedded agency assumes that any entity must be part of its environment\footnote{Embedded agency raises numerous other challenges, such as the uncomputability of Bayes rule proved by \cite{solomonoff09}, or its computational hardness (see \cite{aaronson12}), as well as the risk of wireheading the reward system (see \cite{everitt18}). On the positive side, it allows improvement by a maintenance system, as discussed in Section \ref{sec:maintenance}.}.
Given embedded agency, it then seems that any entity can be analyzed by some other entity \OM{} of the environment. \OM{} can then exploit its analysis to make counterfactual optimization impossible, as proved by Theorem \ref{th:impossibility}.

% Note also that embedded agency raises other fundamental challenges, like the intractability of Bayes rule. In fact, \cite{solomonoff09}'s uncomputability theorem more generally shows that the {\it completeness} of an entity, i.e. its ability to provably ``figure out'' its environment, depends on

\subsubsection{A gap in decision theory}
\label{sec:decision_algorithm}

Perhaps the most important take-away of Theorem \ref{th:impossibility} is that counterfactual optimization cannot be a decision algorithm. Rather, it should be regarded as a property that no decision algorithm always satisfy. At best, a decision algorithm should be designed to {\it often} satisfy counterfactual optimization. However, the more general question of what ought to be expected from a ``good'' decision algorithm seems still far from being resolved.

Theorem \ref{th:impossibility} may share similarities with the no-free-lunch theorems in learning theory (\cite{wolpert96,joyce18}). These theorems suggest that there is no canonical property that identifies the ``good'' learning algorithms, which are arguably what a ``good'' epistemic system should implement. Instead, the defense of Bayesianism, like in \cite{hoang_bayes}, rests upon a myriad a desirable properties, some of which can be proved to be unique to Bayesianism, such as robustness to Dutch book arguments (\cite{teller73,skyrms87}), compatibility with logic (\cite{cox46,cox63,jaynes03}) and statistical admissibility (\cite{wald47,robert07}).

In fact, just as Bayesianism actually includes a large family of epistemologies, each derived from a given prior, the right path forward in decision theory might consist of proving that some desirable property can only be satisfied by decision algorithms taken from a certain family. Unfortunately, this research direction is out of the scope of the present paper.

\section{Newcomb's paradox}
\label{sec:newcomb}

In this section, we analyze Newcomb's paradox under the lens of Bayesian counterfactuals. But first, we need to consider a slight generalization of Newcomb's paradox, which more adequately fits the Bayesian framework.

\subsection{A slight generalization of Newcomb's problem}
\label{sec:problem}

In this paper, we consider a slight variation on Newcomb's problem to make it more realistic. In particular, we will care about the {\it data} that enables the predictor to make its prediction. Let us describe the problem we consider.

\alice{} enters a room with two boxes A and B.
\begin{itemize}
  \item \boxA{} is opaque. \alice{} cannot see what is inside. But she knows how the content of \boxA{} was decided, which is discussed below.
  \item \boxB{} is transparent. \alice{} sees that \boxB{} contains a reward $r>0$.
\end{itemize}
\alice{} is then told that she must decide between two strategies.
\begin{itemize}
  \item The \onebox{} strategy consists of only taking \boxA{}.
  \item The \twobox{} strategy consists of taking both \boxA{} and \boxB{}.
\end{itemize}
What makes Newcomb's paradox interesting is the way the content of \boxA{} is decided.
% Some other entity \OM{} decides whether to put a large reward $R$ in \boxA{}, or to leave \boxA{} empty. To fix ideas, you can imagine $r=\$1,000$ and $R=\$1,000,000$, though our analysis will hold for any other values of the rewards $r$ and $R$.
In some classical versions of Newcomb's paradox, some omniscient entity \OM{} makes a prediction. If \OM{} predicts that \alice{} will \onebox{}, then \OM{} puts a large reward\footnote{To fix ideas, you can imagine $r=\$1,000$ and $R=\$1,000,000$, though our analysis will hold for any other values of the rewards $r$ and $R$.} $R$ in \boxA{}. Otherwise, if \OM{} predicts that \alice{} will \twobox{}, then \OM{} leaves \boxA{} empty. 

However, the omniscience assumption is reasonably criticized for being too unrealistic. % This might give the impression that Newcomb's paradox is inapplicable in practice.
In this paper, we will instead assume that \OM{} is an information processing system which exploits huge amounts of data about \alice{} and about the world. Based on this large database $D_\Omega$, \OM{} infers a probability $\omega \triangleq \pb{\onebox{} \st D_\Omega}$ that \alice{} will \onebox{}. Now, given this probability guess $\omega$, \OM{} will throw a biased coin, which has a probability $\omega$ to land on heads. If the coin lands on heads, \OM{} will put the reward $R$ in \boxA{}. Otherwise, \OM{} leaves \boxA{} empty.

Note that the classical version of Newcomb's paradox is retrieved by assuming that \OM{} knows \alice{}'s decision process and all the inputs to this process. Indeed, \OM{} can then simply simulate this decision process to determine \alice{}'s decision. Our variant is thus a natural generalization of the classical paradox.

Now, \alice{} knows how \OM{} operates. She knows that \OM{} applies Bayes rule and she knows that \OM{} has processed a huge amount of data that \alice{} cannot access. \alice{} wants to maximize her expected (counterfactual) rewards. \alice{} now has to choose between the \onebox{} and the \twobox{} strategy. What should she do?

\subsubsection{Can \OM{} really know more?}
\label{sec:realism}

The odd feature of Newcomb's paradox is the assumption that some external observer \OM{} can better know \alice{} than she knows herself. At first sight at least, it may seem that no one can better guess what \alice{} will decide than \alice{} herself, right before she makes her decision.

However, \cite{libet93}'s famous experiment suggests that this may not be the case. An algorithm processing some magnetic resonance imaging of \alice{}'s brain may then be better able to predict some of \alice{}'s decision — at least a fraction of a second before the decision is made. More generally, it seems in fact clear that humans' decision systems are largely unknown to them. Similarly, and perhaps more strikingly, large institutions often hire external auditing companies to better understand how their decisions are made. Likewise, understanding in details what decisions were made by the YouTube recommender systems, and how they were made, is arguably largely out of reach of any data analytics systems within YouTube.

Perhaps a more realistic way for \OM{} to predict \alice{}'s decision is to collect large amounts of data, not only about \alice{}, but also about entities similar to \alice{}. Typically, a philosopher who posed Newcomb's paradox to generations of students might have gained a remarkable capability to predict their current students' intuitions for the Newcomb's paradox. Similarly, by analyzing all sorts of data available on social medias, an algorithm could detect patterns that would allow it to reliably predict what a given social media user may decide. In fact, \cite{wang18} designed one such algorithm, then replicated by \cite{leuner19}, that achieved superhuman performances at predicting sexual orientations from human faces.

% \input{bayesianism}
% \input{logic}
% \section{Transparent randomized decision systems}
% \input{transparent}
% \input{zero-surprise}
% \input{equal-surprise}
% \section{Resolution}
% \input{emma}
\subsection{The main result}
\label{sec:omega}

Let us focus on the case where \emma{} believes that \OM{} has the same prior as her, and has a supset of her data.
% This means that all the pure Bayesians of our discussion already know \emma{}'s data $D_{\bf E}$. As a result, to simplify notations and without loss of generality, we can assume that $D_{\bf E}$ is empty, or, equivalently, that all the probabilities and expectations that we will compute are actually conditioned on $D_{\bf E}$, even if we do not write it explicitly.
We denote $p \triangleq \pb{\onebox{}}$ \emma{}'s prior on \onebox{}, and $\sigma^2 \triangleq \var{\omega}$ her prior variance on \OM{}'s prediction. Assuming that \emma{} has imperfect knowledge of \dan{} then corresponds to $0<p<1$.
Remarkably, under our assumptions, we can determine \emma{}'s counterfactual preferences based solely on the variables $p$ and $\sigma^2$.
% This will yield a very insightful understanding of Newcomb's paradox. In fact, characterizing \emma{}'s counterfactual preferences is the main result of this paper.

\begin{theorem}
\label{th:main}
  Assume that \emma{} has imperfect knowledge of \dan{}. Suppose also that \emma{} believes that \OM{} is Bayesian, has the same prior as her and knows a supset of her data. Then, \emma{} counterfactually prefers to \onebox{}, if and only if, we have
  $\displaystyle \frac{r}{R} \leq \frac{\sigma^2}{p (1-p)}$.
\end{theorem}

Unsurprisingly, the larger $R$ is relatively to $r$, the more \onebox{} will tend to be preferable. But perhaps what's more interesting is the right-hand side. Note that the quantity $p(1-p)$ is essentially the variance according to \emma{}'s prior on what \dan{} will do. Thus, the right-hand side compares the variance of \OM{}'s prediction to the variance of \dan{}'s decision. Or, put differently, it is a measure of how much \emma{} expects \OM{} to know more than her about \dan{}. Intuitively, the more \emma{} knows about \dan{}, the more she will counterfactually prefer to \twobox{}. But the more she feels that \OM{} knows more than her, the more she will counterfactually prefer to \onebox{}.

Thus, at its core, Newcomb's paradox really seems to be about how much an entity believes that an observer can better know them than they know themselves. This may seem extremely confusing at first glance. But separating our epistemic system from our decision system arguably helps to clarify the origin of the paradox. Indeed, Newcomb's paradox is actually rather about how much the observer can better predict the output of the entity's decision system, than what the entity's epistemic system can predict. In particular, if this epistemic system is unreliable, then it may not be that hard to outperform it.

% In the case of humans, machines and institutions, it seems quite sensible to consider that essentially all entities' epistemic systems, if they exist, do not fully know their own decision systems. We humans do not seem to always know {\it how} we decided what we decided to do. Similarly, and perhaps more strikingly, large institutions often hire external auditing companies to better understand how their decisions are made. Likewise, understanding in details what decisions were made by the YouTube recommender systems, and how they were made, is arguably largely out of reach of any data analytics systems within YouTube.

% In this section, we resolve Newcomb's paradox in the cases where \emma{} knows more than \OM{}, and in the case where \OM{} knows more than \emma{}. In both cases, we also assume that \OM{} and \emma{} have a common prior. This allows us to exploit the Bayesian theorem of the argument of authority, which greatly simplify the analysis.

\subsubsection{Argument of authority}
\label{sec:authority}

Before proving Theorem \ref{th:main}, let us first note the following lemma, which will be critical in the analysis of Newcomb's problem. The lemma states the validity of the argument of authority, if the authority is more a knowledgeable honest Bayesian.

% In this paper, we will restrict ourselves to simple cases of Newcomb's paradox where either \emma{} or \OM{} knows more than the other. This will allow us to apply the Bayesian theorem of the argument of authority.
%
% \begin{definition}
%   An entity knows more than another, if the first entity has collected a supset of the other entity's data, and both know it.
%   % Formally, denoting $\left.\mathcal D\right|_{D}$ the data space still consistent with $D$, $D$ the data of the first agent, and $D'$ the data of the second agent $B$, then $A$ is said to know more than $B$, if and only if, $\left.\mathcal D\right|_{D_A} \subset \left.\mathcal D\right|_{D_B}$.
% \end{definition}

\begin{lemma}[Argument of authority]
\label{lemma:authority}
  If \emma{} believes that \OM{} is an honest Bayesian, that \OM{} has the same prior as her, and that \OM{} has a supset of her data, then \emma{} should believe whatever \OM{} says. More formally, we have
  \begin{equation}
    \pb{T \st \pb{T \st D_\Omega} =p } = p.
     % \quad\text{and}\quad \expect{X \st \expect{X \st D} =x} = x.
  \end{equation}
\end{lemma}

\begin{proof}
  To clarify, denote $\mathcal D_p$ the set of data $D_\Omega$ such that $\pb{T \st D_\Omega} = p$. Then $\mathcal D_p$ excludes all data which do not satisfy this. By the law of total probability, we see that $\pb{T \st \mathcal D_p}$ is then necessarily an average of terms $\pb{T \st D_\Omega}$, for $D_\Omega \in \mathcal D_p$. But all such terms equal $p$, hence the theorem.
\end{proof}

\subsubsection{Useful lemmas}
\label{sec:lemmas}

The proof of Theorem \ref{th:main} yields some insights into the mechanisms at play. In fact, it rests on the following four interesting observations, on what \emma{} predicts based on her uncertainty about \OM{}'s prediction. All the lemmas implicitly make the same assumptions as Theorem \ref{th:main}.

\begin{lemma}
\label{lemma:one-box}
  \emma{}'s prior $p$ on \dan{} deciding to \onebox{} is equal to the expectation of her prior on \OM{}'s prediction, i.e. $p = \expect{\omega}$.
\end{lemma}

\begin{proof}
  According to the law of total probability, $\pb{\onebox{}} = \expectVariable{\omega}{\pb{\onebox{} \st \omega}}$.
  By Lemma \ref{lemma:authority}, if \emma{} learned \OM{}'s prediction $\omega = \pb{\onebox{} \st D_\Omega}$, then she too would assign a probability $\omega$ to \dan{} deciding to \onebox{}.
  Therefore, $\pb{\onebox{} \st \omega} = \omega$. Thus, $p = \pb{\onebox{}} = \expect{\omega}$.
\end{proof}

\begin{lemma}
\label{lemma:boxA}
  \emma{} assigns a prior probability $p$ to \boxA{} containing the large reward $R$, $\pb{A=R} = p$.
\end{lemma}

\begin{proof}
  According to the law of total probability, $\pb{A=R} = \expectVariable{\omega}{\pb{A=R \st \omega}}$. By the decision algorithm of \OM{} (see Section \ref{sec:problem}), we have $\pb{A=R \st \omega} = \omega$. Thus $\pb{A=R} = \expect{\omega} = p$.
\end{proof}

\begin{lemma}
\label{lemma:counterfactual-onebox}
  \emma{}'s posterior belief that \boxA{} contains the large reward $R$, given a \onebox{} decision by \dan{}, is given by $\pb{A=R \st \onebox{}} = p + \displaystyle \frac{\sigma^2}{p}$. Interestingly, this quantity is larger than the prior probability, especially if the variance $\sigma^2$ on \OM{}'s prediction is large and if \OM{} is expected to essentially predict \twobox{}..
\end{lemma}

\begin{proof}
  By the law of total probabilities, we have
  \begin{align}
    &\pb{A=R \st \onebox{}}
    % = \expectVariable{\omega}{\pb{A=R \st \onebox{} \wedge \omega} \st \onebox{}} \\
    = \sum_{\omega} \pb{A=R \st \onebox{} \wedge \omega} \pb{\omega \st \onebox{}} \\
    &\qquad = \sum_{\omega} \frac{\pb{A=R \wedge \onebox{} \st \omega}}{\pb{\onebox{} \st \omega}} \frac{\pb{\onebox{} \st \omega} \pb{\omega}}{\pb{\onebox{}}} \\
    &\qquad = \sum_{\omega} \pb{A=R \st \omega} \pb{\onebox{} \st \omega} \frac{\pb{\omega}}{p} \\
    &\qquad = \frac{1}{p} \sum_\omega \omega^2 \pb{\omega} = \frac{\expect{\omega^2}}{p}
    = \frac{p^2 + \sigma^2}{p} = p + \frac{\sigma^2}{p}.
  \end{align}
  In the second line, we used Bayes rule. In the third line, we use the conditionally independence of events $A=R$ and \onebox{}, given $\omega$, as well as Lemma \ref{lemma:one-box}. In the fourth line, we exploited the way the decision of $A=R$ is decided by $\omega$, as well as the Lemma \ref{lemma:authority} which implies $\pb{\onebox{} \st \omega}= \omega$. The last line also uses the definition of $\sigma^2$.
\end{proof}

\begin{lemma}
\label{lemma:counterfactual-twobox}
  \emma{}'s posterior belief that \boxA{} contains the large reward $R$, given a \twobox{} decision by \dan{}, is given by $\pb{A=R \st \twobox{}} = p - \displaystyle \frac{\sigma^2}{(1-p)}$. Perhaps not surprisingly given Lemma \ref{lemma:counterfactual-onebox}, this posterior is smaller than the prior, especially for large values of $\sigma^2$ and when \OM{} is expected to essentially predict \onebox{}.
\end{lemma}

\begin{proof}
Recall that, by Lemmas \ref{lemma:one-box} and \ref{lemma:boxA}, we have $\pb{A=R} = p = \pb{\twobox{}}$. Now, Bayes rule yields
\begin{align}
  &\pb{A=R \st \twobox{}} = \frac{\pb{\twobox{} \st A=R} \pb{A=R}}{ \pb{\twobox{}}} \\
  &\qquad = \frac{\left(1- \pb{\onebox{} \st A=R}\right) p}{1-\pb{\onebox{}}} \\
  &\qquad = \left(1 - \frac{\pb{A=R \st \onebox{}} \pb{\onebox{}}}{\pb{A=R}}\right) \frac{p}{1-p} \\
  &\qquad = \left(1-p - \frac{\sigma^2}{p} \right) \frac{p}{1-p}
  = 1- \frac{\sigma^2}{1-p},
\end{align}
where, in the fourth line, we used Lemma \ref{lemma:counterfactual-onebox}.
\end{proof}

\subsubsection{Proof of the main theorem}
\label{sec:proof_main}

Theorem \ref{th:main} then follows from our four previous lemmas.

\begin{proof}[Proof of Theorem \ref{th:main}]
  From Lemmas \ref{lemma:counterfactual-onebox} and \ref{lemma:counterfactual-twobox}, we have
  \begin{align}
    &\expect{A \st \onebox{}} = R \cdot \pb{A=R \st \onebox{}} = \left( p + \frac{\sigma^2}{p} \right) \bar R, \quad \text{and} \\
    &\expect{A+B \st \twobox{}} = R \left( p - \frac{\sigma^2}{1-p} \right) + r.
  \end{align}
  Now, \emma{} counterfactually prefers to \onebox{} if and only if the former conditional expectation is larger than the second, which is equivalent to saying $\displaystyle \frac{\sigma^2}{p} \geq \frac{r}{R} - \frac{\sigma^2}{1-p}$. Rearranging the terms yields the theorem.
\end{proof}

\subsection{Corollaries}
\label{sec:corollaries}

In this section, we discuss the corollaries of Theorem \ref{th:main}. These corollaries yield greater insight into the key variables of Newcomb's paradox.

\subsubsection{If \emma{} thinks she knows as much}
\label{sec:twobox}

\begin{corollary}
\label{cor:twobox}
  For any rewards $R>r$, if \emma{} has imperfect knowledge of \dan{} and if she expects \OM{} to have the same prior and data as her, then she will counterfactually prefer to \twobox{}.
\end{corollary}

\begin{proof}
  This is the special case where $\sigma^2 = 0$.
\end{proof}

In particular, if \emma{} is almost sure than \OM{} has the same prior and data as her, and if she is almost sure that \dan{} will \twobox{}, then she is almost sure that \dan{}'s decision is counterfactually optimal.

\subsubsection{The more \OM{} knows, the better \onebox{} looks}
\label{sec:data}

In this section, we show that $\sigma^2$ is an increasing function of the size of the data $D_\Omega$ that \emma{} suspects \OM{} to have. In other words, the more \emma{} thinks that \OM{} has a more data than her, the more she will lean towards counterfactually preferring to \onebox{}.

\begin{lemma}
\label{lemma:variance}
  Denote $\omega(D) = \pb{\onebox{} \st D}$ the prediction made based on data $D$ about \dan{}. If we know with that the data $D_\Omega^+$ contain data $D_\Omega$, even though both are unknown, then
  \begin{equation}
    \var{\omega(D_\Omega^+)} = \var{\omega(D_\Omega)} + \expectVariable{D_\Omega}{\var{\omega(D_\Omega^+) \st D_\Omega}}.
  \end{equation}
\end{lemma}

\begin{proof}
  Note that
  \begin{align}
    \var{\omega(D_\Omega^+)}
    &= \expect{\omega(D_\Omega^+)^2} - p^2
    = \expectVariable{D_\Omega}{\expectVariable{D_\Omega^+}{\omega(D_\Omega^+)^2 \st D_\Omega}} - p^2 \\
    &= \expectVariable{D_\Omega}{\expectVariable{D_\Omega^+}{\omega(D_\Omega^+)^2 \st D_\Omega} - \omega(D_\Omega)^2 + \omega(D_\Omega)^2} - p^2 \\
    &= \expectVariable{D_\Omega}{\expectVariable{D_\Omega^+}{\omega(D_\Omega^+)^2 \st D_\Omega} - \expectVariable{D_\Omega^+}{\omega(D_\Omega^+) \st D_\Omega}^2} \nonumber \\
    &\qquad \qquad \qquad \qquad \qquad \qquad + \expect{\omega(D_\Omega)^2} - \expect{\omega(D_\Omega)}^2 \\
    &= \expectVariable{D_\Omega}{\var{\omega(D_\Omega^+) \st D_\Omega}} + \var{\omega(D_\Omega)},
  \end{align}
  where, in the third line, we used the fact that, if $D_\Omega^+$ contains $D_\Omega$, then $\omega(D_\Omega) = \expect{\omega(D_\Omega^+) \st D_\Omega}$.
\end{proof}

As a corollary, we have $\var{\omega(D_\Omega^+)} \geq \var{\omega(D_\Omega)}$. This means that, for a Newcomb paradox with \OM{}$^+$ that has more data than \OM{}, the variance $\sigma^2$ of the predictor's prediction is larger. In particular, as \OM{} collects more and more data, \emma{} will tend more and more towards \onebox{}.

\subsubsection{If \boxB{} contains more rewards ($r \geq R$)}
\label{sec:reward_inconsistency}

Unsurprisingly, if the reward $r$ of \boxB{} is necessarily larger than the reward of \boxA{}, then \twobox{} is counterfactually preferable.

\begin{corollary}
\label{cor:r_geq_R}
  If $r \geq R$, if \emma{} has imperfect knowledge of \dan{}, and if she expects \OM{} to have the same prior and more data, then \emma{} counterfactually prefers to \twobox{}.
\end{corollary}

\begin{proof}[Proof of Corollary \ref{cor:r_geq_R}]
  Lemma \ref{lemma:variance} with $D_\Omega^+ = D_\Omega \wedge \onebox{}$ implies that $p(1-p) \geq \sigma^2$. If $r \geq R$, then we have $\displaystyle \frac{r}{R} \geq 1 \geq \frac{\sigma^2}{p(1-p)}$.
\end{proof}

\subsubsection{Quasi-omniscient Omega}
\label{sec:omniscient}

One particularly interesting particular case is that of a quasi-omniscient \OM{}, which knows {\it a lot more} about \dan{} than \emma{} does. We formalize quasi-omniscience as follows.

\begin{definition}
  \emma{} believes \OM{} to be quasi-omniscient about \dan{}, if \emma{} rules out an uncertain prediction by \OM{}. Formally, \OM{} is $\delta$-omniscient if $\pb{\omega \in (\delta, 1-\delta)} = 0$.
\end{definition}

Note that $\delta$-omniscience is a property that depends on \emma{}'s prior on the data $D_\Omega$ that \OM{} has access to. In fact, it is arguably {\it not} a fundamental property of \OM{} itself.

One thing that makes this case interesting is that, for a fixed prior $p$, especially in the limit $\delta \rightarrow 0$, the variance $\sigma^2$ gets maximized. In particular, throughout this section, we assume $0 < \delta < \min\{ p, 1-p\}$, which allows to guarantee the fact that the conditional expectations that we will consider are well-defined.

\begin{lemma}
  If \emma{} believes that \OM{} knows more and is $\delta$-omniscient, then $\sigma^2 \geq p(1-p) - 3\delta - o(\delta)$.
\end{lemma}

\begin{proof}
  First, recall that, by Lemma \ref{lemma:one-box}, $\expect{\omega} = p$. Now, denote $q \triangleq \pb{\omega \geq 1-\delta}$. By the law of total probability,
  \begin{align}
    p
    = \expect{\omega \st \omega \leq \delta} (1-q) + \expect{\omega \st \omega \geq 1- \delta} q
    % &= \expect{\omega \st \omega \leq \delta} + \left( \expect{\omega \st \omega \leq 1-\delta} - \expect{\omega \st \omega \leq \delta} \right) q \\
    \leq \delta + q,
  \end{align}
  and thus $q \geq p-\delta$. Therefore, $\sigma^2 = \expect{\omega^2} - p^2$. Now, note that
  \begin{align}
    \expect{\omega^2}
    &= \expect{\omega^2 \st \omega \leq \delta} \pb{\omega \leq \delta} + \expect{\omega^2 \st \omega \geq 1- \delta} \pb{\omega \geq 1- \delta} \\
    &\geq 0 + (1-\delta)^2 q \geq (1-\delta)^2 (p-\delta) \geq p - (1+2p)\delta - o(\delta).
  \end{align}
  Therefore, we have $\sigma^2 \geq p(1-p) - 3\delta - o(\delta)$.
\end{proof}

% \begin{lemma}
%   Denoting $q = \pb{\omega=1-\delta}$, we must have $\bar \omega = q + \delta (1- 2q)$, and $\sigma^2 = q(1-q) (1-2\delta)^2$.
% \end{lemma}
%
% \begin{proof}
%   This follows from noting that $\omega$ can be written $\delta + (1-2\delta) U$, where $U \in \{0,1\}$ and $\pb{U = 1} = p$. The mean and the variance of $U$ are well-known to be $p$ and $p(1-p)$.
% \end{proof}

We can then state the following insightful corollary of Theorem \ref{th:main}.

\begin{corollary}
\label{cor:omniscient}
  If \emma{} has imperfect information about \dan{}, if she believes \OM{} knows more than her, and if she believes \OM{} to be $\delta$-omniscient, then, in the limit $\delta \ll \min\{p,1-p\}$, \emma{} counterfactually prefers to \onebox{}.
\end{corollary}

\begin{proof}
  Algebraic manipulations show that $\frac{\sigma^2}{p(1-p)} = 1 - \frac{3\delta}{p (1-p)} - o\left(\frac{\delta}{p (1-p)}\right)$. We conclude by applying Theorem \ref{th:main}.
\end{proof}

In other words, for any given values of $R > r$, and for any given prior uncertainty from \emma{} about \dan{}'s decision, there is a sufficiently omniscient \OM{} such that \emma{} will counterfactually prefer to \onebox{}. Or, put differently, if \alice{} is pretty sure that she will \onebox{}, and if she is pretty sure that \OM{} is pretty-pretty-pretty sure of what \alice{} decide, then \alice{} will be pretty sure to make a counterfactually optimal decision.

This conclusion may clarify some of the disagreements about Newcomb's paradox. If one believes that \alice{} fully knows what she decides, then it seems meaningless to consider that \OM{} could know a lot better than \alice{} what she will decide. Similarly, by considering that \OM{} is fully omniscient, one may be tempted to exclude some scenarios, such as \OM{} being wrong. But such scenarios are arguably not quite Bayesian, nor realistic. By more carefully considering imperfect knowledge, purely Bayesian counterfactuals in Newcomb-like problems arguably seem less mysterious.

% In particular, if \emma{} $\varepsilon$-nearly knows that \dan{} counterfactually optimize, then there are settings where \emma{} can be expecting \alice{} to \onebox{}.

% \begin{corollary}
%   Consider any rewards $R>r>0$. Then, if \alice{} is a pure Bayesian who $\varepsilon$-nearly knowingly counterfactually optimizes, there exists $\delta>0$ such that if \OM{} is any $\delta$-nearly omniscient and if \alice{} knows it, then \alice{} will \onebox{} with probability $1-\sfrac{\varepsilon}{2}$.
% \end{corollary}

% \input{pregame}
\section{Discussion}
\label{sec:discussion}

While our theorems closed a few problems, they arguably raised even more questions. In this section, we discuss take-aways and future research directions.

\subsection{What if \OM{} does not know more?}

The results of Section \ref{sec:newcomb} assumed that \emma{} believes that \OM{} has the same prior as her, as well as a supset of her data. In this section, we discuss the challenges to remove this assumption.

\subsubsection{If \emma{} knows more}
\label{sec:irreducible}

If \emma{} has a supset of \OM{}'s data, then she would know \OM{}'s prediction $\omega$. It can be shown that, as a result, if \emma{} has imperfect knowledge of \dan{}, then so does \OM{}. But this also implies that \OM{}'s actual decision to put $R$ in \boxA{} is not fully known to \OM{}. It may thus also be imperfectly known to \emma{}.

Unfortunately, this last bit of uncertainty may be where a Newcomb-like paradox kicks in. Typically, \OM{}'s coin toss might be further biased by another entity \OM{}$^+$ that does have more data than \emma{}\footnote{Recall that the probabilities we consider here are Bayesian. They describe an entity's ignorance. Therefore, they are not assumed to describe some fundamental randomness within the laws of physics.}.

Note though that if \emma{} believes that the coin toss is independent from \dan{}'s decision given \emma{}'s data, then \emma{} can be proved to counterfactually prefer to \twobox{}. Similarly, if \OM{}'s uncertainty on the content of \boxA{} disappears, for instance if \OM{} now decides to put $R$ in \boxA{} if and only if its credence in \onebox{} satisfies $\pb{\onebox{} \st D_\Omega} \geq 1/2$, then \emma{} will also counterfactually prefer to \twobox{}.

\subsubsection{No one has more data}
\label{sec:agreement}

The most general and realistic case is evidently when \emma{} and \OM{} have access to very different data. In such a case, Lemma \ref{lemma:authority} does not apply. Unfortunately, as a result, \emma{}'s counterfactual preferences are then much harder to analyze.

Nevertheless, it seems that the global take-away of our analysis still applies. Essentially, the more data \emma{} has compared to \OM{}, the more \twobox{} probably seems counterfactually preferable to her. Conversely, the more data \OM{} has, the more it seems that \emma{} will counterfactually prefer to \onebox{}. However, we leave open the problem of proving this mathematically.

\subsubsection{Diverging priors}
\label{sec:priors}

The other important assumption we have made throughout our analysis is that \emma{} believes that \OM{} has the same Bayesian prior as her. One way to weaken this assumption is to suppose instead that \emma{} has a prior on \OM{}'s possible priors. By expliciting the prior on priors, it is then only a computation for \emma{} to determine what decision is counterfactually preferable.

Again, it seems intuitive that, as \OM{} gathers more and more data, \emma{} will tend to counterfactually prefer to \onebox{}. We conjecture that, even in this case, if \emma{} assigns a strictly positive probability to \OM{}'s actual prior, and if she believes \OM{} to be sufficiently quasi-omniscient, then she will counterfactually prefer to \onebox{}. This has yet to be proved though.
% But proving this result is left as an open problem.

% Typically, in \cite{solomonoff64a,solomonoff64b}, the prior depends on some reference universal Turing machine. Interestingly though, by the universality theorem of Turing machines \cite{turing36}, this means that any prior defined by \cite{solomonoff64a,solomonoff64b} for \emma{} gives a strictly positive probability to any other prior defined by \cite{solomonoff64a,solomonoff64b} for \OM{}. In fact, \OM{}'s prior can be derived from \emma{}'s by giving \emma{} a description of the Turing machine used by \OM{}.
% Thus, there exists some data $\Omega$ such that, for any event $X$, we have $\pbVariable{\bf E}{X \st \Omega} = \pbVariable{\Omega}{X}$, where $\mathbb P_{\bf E}$ describes \emma{}'s credences and $\mathbb P_\Omega$ describes \OM{}'s.

% So, assuming that \emma{} has a prior on \OM{}'s prior, her estimation of \OM{}'s prediction is equivalently described by the case where \emma{} and \OM{} have the same prior, and where \emma{} has some uncertainty about some data $\Omega$ used by \OM{}. In particular, by denoting $\omega \triangleq \pbVariable{\bf E}{\onebox{} \st D_\Omega \wedge \Omega}$ and $\sigma^2 \triangleq \varVariable{D_\Omega \wedge \Omega}{\omega}$, Theorem \ref{th:main} still applies.

\subsection{Practical lessons}
\label{sec:practical}

Our generalization of Newcomb's paradox arguably made it more realistic. However, there still seems to be a gap between our analysis and practical Newcomb-like problems.

\subsubsection{Newcomb-like pragmatic problems}

% Newcomb's paradox is sometimes criticized for being unrealistic. Our analysis may help gain insight into its crux, which may then help understand why it probably already applies to many everyday problems.

At its heart, Newcomb's paradox is about the possibility for an external observer to better predict what we will decide than what we can predict ourselves. Crucially, we showed that such an external observer does not need to be some supernatural omniscient entity. Any observer with vastly more data than us could potentially predict our future decisions better than we can — and use this to trick us.

Arguably, billions of dollars are currently being invested to create such algorithmic observers. By leveraging the huge amounts of data provided by their users, social media algorithms are constantly trying to predict how likely we are to click on the contents that they decide to recommend to us. Meanwhile, the way we use these social medias, often without paying our utmost attention to our decision systems, means that we may make decisions without even realizing it. In such a case, it seems reasonable to argue that algorithms actually already know many of our future decisions better than we do ourselves.

It may be interesting for future work to investigate how Newcomb's paradox can inform us on what we ought to do in such contexts; or at least, on what decisions would be least counterfactually wrong.

\subsubsection{Remarks for non-Bayesians}
\label{sec:non-bayesian}

We assumed that \emma{} and \OM{} are pure Bayesians. This hypothesis turned out to be extremely useful to gain insights into Newcomb's paradox. However, it is noteworthy that, in general, Bayesianism requires unreasonable computing resources, and thus cannot be applied exactly in practice, as explained by \cite{solomonoff09}. This leaves us with the question of the robustness of our results, if we now consider non-Bayesian entities.

Note that the assumption that \OM{} is (believed to be) Bayesian is not critical.
% In fact, what we assumed was that \emma{} believed that \OM{} was a Bayesian.
The critical feature of our analysis is rather that \OM{} is expected to have more data than \emma{}, which allows \OM{} to better predict what \dan{} will decide. In fact, it suffices that \OM{}'s prediction $\omega$ remains strongly correlated with \dan{}'s decision, even given \emma{}'s data.
% For Newcomb's paradox to apply, it suffices that \OM{} has access to such data, and that the way it processes such data is reasonably correct.

The case of \emma{} is slightly more problematic, as \emma{} needs to apply the law of probabilities to compute counterfactuals. However, again, by assuming that she can reasonably well approximate the computation of the counterfactuals, then it is sensible to consider that she is engaging in approximate Bayesian counterfactual reasoning. The general take-aways of our analysis would then apply.

\subsubsection{Logical non-omnniscience}
\label{sec:complexity}

% There is a fundamental reason why practical entities cannot be Bayesian. while Bayesianism has many desirable properties, it is computationally intractable, if not uncomputable (see \cite{solomonoff09}). In practice, because of computational complexity constraints, epistemic systems cannot be Bayesian.

Even if epistemic systems know exactly their decision algorithms and the inputs of these algorithms, they may be incapable to derive the decision unless they perform themselves the computations of the decision algorithms. This postulate was called {\it computational irreducibility} by \cite{wolfram02}. A consequence of this postulate is that computational limits actually add another source of uncertainty, which cannot be taken into account by the Bayesian framework alone, as explained by  \cite{hoang_bayes}.

More generally, as argued by \cite{aaronson12}, computational complexity theory may be critical to understand diverse philosophical paradoxes. Its implications to Newcomb's paradox have yet to be analyzed.

\subsection{Deconstructing entities}

The core idea of our analysis was a clear separation between an entity's epistemic system and their decision system. In particular, we exploited the fact that the epistemic system should have some epistemic uncertainty about their decision system, which we argued to be critical to understand Newcomb's paradox. However, the decomposition of any entity into different systems with different tasks can be taken further, as argued by \cite{Hoang19}.

\subsubsection{Data collection system}
\label{sec:data}

Interestingly, our analysis highlights the importance of data in Newcomb's paradox. Indeed, we saw that \emma{} will consider that \onebox{} is counterfactually preferable if and only if she expects \OM{} to have access to sufficiently more data $D_\Omega$ so that \OM{}'s uncertainty about \dan{}'s decision is greatly reduced. Note that, however, we did not discuss {\it how} \OM{} would access such data.

More generally, data is arguably critical for any entity, both for the epistemic and the decision systems. Therefore, it seems critical for the entity, or for any analysis of the entity, to carefully design or understand their {\it data collection system}. Arguably, in the case of the COVID-19 crisis, in at least some parts of the world, this system was defective, typically because of a lack of tests or because of data misreporting. In practice, improving data collection systems is a key aspect of improving an entity.

In particular, especially for large-scale critical applications, designing a quality data collection systems seems to require features such as {\it data authentication}, {\it data storage} and {\it data communication}. Results in the fields of cryptography, differential privacy or distributed computing may add insights into Newcomb-like problems.

\subsubsection{Reward systems}
\label{sec:reward}

In this paper, we assumed that the rewards were given by \OM{}. In fact, \OM{} was essentially \emma{}'s reward system.

In practice, most entities' rewards are much more complex to compute. Humans may care about happiness, glory or flourishing, or a certain combination of all three. Organizations may want to advance their cause, publish papers, protect their nations, save lives or make their business sustainable. Finally, many algorithms have been designed to maximize profits, clicks or user attention.

In most cases, rewards are arguably what determine these entities' decisions the most, and thus their impacts on the world. In particular, a watch-time maximizing algorithm may flood the Internet with clickbait, virulent and unreliable videos, which may then encourage climate denialism \cite{allgaier19}, promote dangerous health recommandations \cite{johnson20} or normalize radicalizations \cite{ribeiro20}, while a profit-maximizing company may disregard its externalities, and a drug addict may neglect the long-term effects of their drug consumption. Note that, in these examples, the rewards to be maximized are not harmful in themselves; but they are oblivious of major negative {\it side effects}. Carefully understanding and designing entities' reward systems seems critical.

One important feature of reward systems is that they too need to rely on a data collection system and an epistemic system. To illustrate, the rewards to be given to a public health organization should arguably depend on the actual health of the population. However, if this population does not get tested, or if the mental health of the population is not properly inferred from the data, then the public health organization's rewards may be misleading. This defect of the computation of the rewards may then create flawed motivations for the entity, which can lead to poor decisions by the decision system, and to poor assessments by the epistemic system.

More generally, as argued by \cite{Hoang19}, the careful understanding and design of entities' reward systems is arguably the most critical feature to understand these entities, and to make them more {\it robustly beneficial} for the future of our planet. Perhaps Newcomb-like paradoxes can shed more light into what can and cannot be achieved from the interaction between the reward system and the other systems of an entity.

\subsubsection{Maintenance system}
\label{sec:maintenance}

One last component discussed in \cite{Hoang19} is the {\it maintenance system}. Given that any entity almost surely has imperfect data collection, reward, epistemic and decision systems, it seems critical that the entity is aware of this, and actively aims to combat the defects of these components.

One example of defect within humans, and arguably within organizations too, is the {\it confirmation bias}. A decision made by our decision systems can sometimes harm the capabilities of our epistemic systems, by irrationally favoring justifications of our decisions, even when they are poor (see the survey by \cite{nickerson98}).

Another example of defect is known as the law of \cite{goodhart75}, which says that ``when a measure becomes a target, it ceases to be a good measure''. Recently, \cite{Elmhamdi20} even showed that a fat-tail discrepancy between a deployed imperfect reward system and its ideal version could make reward maximization {\it infinitely bad}, according to the ideal reward system.

It thus seems critical for any {\it robustly beneficial} entity evolving in complex environments with feedback loops to have a maintenance system that keeps track of the imperfections of their epistemic, decision, data collection and reward system, and that actively tries to fix and improve these systems if needed.

\section{Conclusion}
\label{sec:conclusion}

This paper proposed to more clearly distinguish an entity's epistemic system from its decision system. I argued that this allowed to better understand Newcomb-like problems. In particular, based on this insight, I proved the impossibility of counterfactual optimization. Moreover, I showed that Newcomb's problem can be reformulated in terms of the capability of some observer to better predict the decision of an entity than the epistemic system of the entity can. Finally, I discussed numerous further research challenges to better understand information processing entities, and to better design entities that make robustly beneficial decisions.

\subsection*{Acknowledgement}

The author is thankful for insightful discussions with El Mahdi El Mhamdi, among others.

\bibliographystyle{plainnat}
\bibliography{references}

\end{document}